\documentclass[letterpaper, 10pt, conference,onecolumn]{ieeeconf}  

\IEEEoverridecommandlockouts                              

\title{\LARGE \bf  A transverse Hamiltonian  variational technique for open quantum stochastic systems and its application to coherent quantum control$^*$}

\author{Igor G. Vladimirov$^{\dagger}$
\thanks{$^*$This work is supported by the Australian Research Council.}
\thanks{$^\dagger$UNSW Canberra, Australia.
{\tt igor.g.vladimirov@gmail.com}.}
}

\usepackage{amsmath}
\usepackage{amssymb}
\usepackage{amsfonts}
\usepackage{graphicx}
\usepackage{color}
\usepackage{mathptmx} 
\usepackage{times}

\usepackage{datetime}

\DeclareMathAlphabet{\bit}{OML}{cmm}{b}{it}

\newtheorem{theorem}{Theorem}

\newtheorem{example}{Example}

\def\ad{\mathrm{ad}}           
\def\<{\leqslant}           
\def\>{\geqslant}           

\def\Re{\mathrm{Re}}   
\def\Im{\mathrm{Im}}   

\def\mR{\mathbb{R}}    
\def\mH{\mathbb{H}}    
\def\cH{\mathcal{H}}    
\def\Tr{\mathrm{Tr}}   
\def\rT{\mathrm{T}}    

\def\bE{\mathbf{E}}    

\def\bra{\langle}
\def\ket{\rangle}

\def\re{\mathrm{e}}    

\def\rd{\mathrm{d}}    

\def\d{\partial}    

\def\x{\times}
\def\ox{\otimes}
\def\op{\oplus}

\def\cZ{\mathcal{Z}}

\def\b1{\mathbf{1}}

\def\bD{\mathbf{D}}

\def\bH{\mathbf{H}}

\def\cF{\mathcal{F}}
\def\cW{\mathcal{W}}

\def\cD{\mathcal{D}}

\def\cG{\mathcal{G}}
\def\cI{\mathcal{I}}

\def\cT{\mathcal{T}}

\def\bG{\mathbf{G}}
\def\bL{\mathbf{L}}

\def\mS{\mathbb{S}}
\def\mA{\mathbb{A}}

\def\eps{\epsilon}
\def\Ups{\Upsilon}

\def\bOmega{\bit{\Omega}}
\def\bTheta{\bit{\Theta}}
\def\bJ{\mathbf{J}}
\def\sj{\mathsf{j}}

\def\fH{\mathfrak{H}}
\def\fF{\mathfrak{F}}
\def\cX{\mathcal{X}}
\begin{document}
\maketitle
\thispagestyle{empty}
\pagestyle{plain}

\begin{abstract}
This paper is concerned with variational methods for  nonlinear open quantum systems  with Markovian dynamics governed by Hudson-Parthasarathy quantum stochastic differential equations. The latter are driven by quantum Wiener processes of the external boson fields and are specified by the system Hamiltonian and system-field coupling operators. We consider the system response to perturbations of these energy operators and introduce a transverse Hamiltonian which encodes the propagation of the perturbations through the unitary system-field evolution. This provides a tool for the infinitesimal perturbation analysis and development of optimality conditions for coherent quantum control problems.  
We apply the transverse Hamiltonian  variational technique to a mean square optimal
coherent quantum filtering problem for a measurement-free  cascade connection of quantum systems.
\end{abstract}

\section{INTRODUCTION}\label{sec:intro}

The Hudson-Parthasarathy quantum stochastic calculus \cite{H_1991,HP_1984,P_1992} provides a unified framework for modelling the 
dynamics of open quantum systems interacting with 
environment. The governing quantum stochastic differential equations (QSDEs) are driven by a quantum counterpart of the classical Wiener process which acts on a boson Fock space and models a heat bath of external fields. An important feature of the QSDEs is that their drift and dispersion terms are specified by the Hamiltonian and coupling operators, which describe the energetics of the system and its interaction with the environment, and also involve a scattering matrix representing the photon exchange between the fields \cite{P_1992}.
These energy operators are usually functions (for example, polynomials or Weyl quantization integrals \cite{F_1989})
of the system variables, and the form of this dependence affects tractability of the resulting quantum system.

In particular, quadratic Hamiltonians and linear system-field coupling operators  lead to a class of open quantum harmonic oscillators (OQHOs) \cite{EB_2005,GZ_2004} which play an important role in the linear quantum control theory \cite{DP_2010,JNP_2008,P_2010}. The linear dynamics of such systems are relatively well understood and are similar to the classical linear SDEs in some regards, such as the preservation of  Gaussian nature of quantum states \cite{JK_1998,KRP_2010} driven by vacuum fields. Nevertheless, the coherent (that is, measurement-free) quantum analogue \cite{NJP_2009} of the classical  LQG control problem \cite{AM_1989,KS_1972} for OQHOs  is complicated by the physical realizability (PR) constraints on the state-space matrices of the fully quantum controller which come from the special structure of QSDEs mentioned above. Similar constraints are present in the coherent quantum filtering (CQF) problem \cite{MJ_2012}.

One of the existing approaches to the coherent quantum LQG (CQLQG) control and filtering problems employs the Frechet derivatives of the LQG cost (with respect to the state-space matrices subject to the PR constraints)
for obtaining the optimality conditions \cite{VP_2013a,VP_2013b} and for the numerical optimization \cite{SVP_2015}. This approach  takes into account the quantum nature of the underlying problem only through the PR constraints, being ``classical'' in all other respects, which has its own advantages from the viewpoint of well-developed optimization methods. However, a disadvantage of this approach lies in the fact that it is limited to certain parametric classes of linear controllers and filters. In particular, the resulting optimality   conditions do not provide information on whether nonlinear quantum controllers or filters can outperform the linear ones for linear quantum plants.

In the present paper, we  outline a fully quantum variational technique which allows the sensitivity of the dynamic and output variables of a nonlinear quantum stochastic system to be investigated with respect to arbitrary (that is, not only linear-quadratic) perturbations of the Hamiltonian and coupling operators. This approach is based on using a transverse Hamiltonian, an auxiliary time-varying operator which encodes the propagation of such perturbations through the unitary system-field evolution.
This leads to an infinitesimal perturbation formula for quantum averaged performance criteria (such as the LQG cost functional) which is applicable to the development of optimality conditions in coherent quantum control problems for larger classes of controllers and filters.
In particular, this approach allows the sensitivity of OQHOs with quadratic performance criteria to be studied with respect to general perturbations of the energy operators. In fact, the transverse Hamiltonian   method has already been employed in \cite{V_2015b} to establish local sufficiency of linear filters in the mean square optimal CQF problem \cite{VP_2013b} for linear plants with respect to varying the energy operators of the filter along Weyl operators \cite{F_1989}.
Note that our approach is different from \cite{GBS_2005} which develops a quantum Hamilton-Jacobi-Bellman principle  for a measurement-based quantum feedback control problem, where the role of state variables is played by the density operator. It is also relevant to mention a parallel between the presented perturbation analysis and the fluctuation-dissipation theorem \cite{LL_1980}.

The paper is organised as follows. Section~\ref{sec:not} outlines notation used in the paper. Section~\ref{sec:QSDE} specifies the class of open quantum systems in Hamiltonian formulation.  Section~\ref{sec:OQHO} discusses sensitivity of OQHOs to parametric perturbations   within quadratic Hamiltonians and linear system-field coupling operators. Section~\ref{sec:TH} introduces the transverse Hamiltonian. The latter is used in   Section~\ref{sec:sys} for perturbation analysis of system operators. Section~\ref{sec:func} extends the transverse  Hamiltonian technique  to quantum averaged performance criteria.  Section~\ref{sec:CQF} applies this approach to a mean square optimal CQF problem. Section~\ref{sec:conc} makes concluding remarks.

\section{PRINCIPAL NOTATION}\label{sec:not}

In what  follows, $[A,B]:= AB-BA$ denotes the commutator of linear operators $A$ and $B$ on a common space. As a linear superoperator, the  commutator with a fixed operator $A$, is denoted by $\ad_A(\cdot):= [A,\cdot]$.
This extends to the commutator $(n\x m)$-matrix  $
    [X,Y^{\rT}]
    :=
    XY^{\rT} - (YX^{\rT})^{\rT} = ([X_j,Y_k])_{1\< j\< n,1\< k\< m}
$ for a vector $X$ of operators $X_1, \ldots, X_n$ and a vector $Y$ of operators $Y_1, \ldots, Y_m$.  Vectors are organized as columns unless indicated otherwise,  and the transpose $(\cdot)^{\rT}$ acts on matrices of operators as if their entries were scalars.
In application to such matrices, $(\cdot)^{\dagger}:= ((\cdot)^{\#})^{\rT}$ denotes the transpose of the entry-wise operator adjoint $(\cdot)^{\#}$. For complex matrices,  $(\cdot)^{\dagger}$ is the usual complex conjugate transpose  $(\cdot)^*:= (\overline{(\cdot)})^{\rT}$.
The subspaces of real symmetric, real antisymmetric and complex Hermitian matrices of order $n$ are denoted by $\mS_n$, $\mA_n$
 and
$
    \mH_n
    :=
    \mS_n + i \mA_n
$, respectively,  where $i:= \sqrt{-1}$ is the imaginary unit. The real and imaginary parts of a complex matrix are denoted by $\Re(\cdot)$ and $\Im(\cdot)$. These extend to matrices $M$ with operator-valued entries as $\Re M = \frac{1}{2}(M+M^{\#})$ and $\Im M = \frac{1}{2i}(M-M^{\#})$ which consist of self-adjoint operators.
Positive (semi-) definiteness of matrices and the corresponding partial ordering  are denoted by ($\succcurlyeq$) $\succ$.  Also, $\mS_n^+$ and $\mH_n^+$ denote the sets of positive semi-definite real symmetric and complex Hermitian matrices of order $n$, respectively.
The tensor product of spaces or operators (in particular, the Kronecker product of matrices) is denoted by $\ox$. The tensor product $A\ox B$ of operators $A$ and $B$ acting on different spaces will sometimes be abbreviated as $AB$.
The identity matrix of order $n$ is denoted by $I_n$, while the identity operator on a space $H$ is denoted by $\cI_H$.
The Frobenius inner product of real or complex matrices  is denoted by
$
    \bra M,N\ket
    :=
    \Tr(M^*N)
$.
The expectation $\bE \xi := \Tr(\rho \xi)$  of a quantum variable $\xi$ over a density operator $\rho$ extends entrywise to vectors and matrices of such variables. 


\section{OPEN QUANTUM STOCHASTIC SYSTEMS BEING CONSIDERED}\label{sec:QSDE}

Consider an open quantum system which interacts with a heat bath of $m$ external boson fields and is endowed with $n$ dynamic variables $X_1(t), \ldots, X_n(t)$ evolving in time $t\> 0$. The system variables are self-adjoint operators on a composite system-field Hilbert space $\cH\ox \cF$, where $\cH$ is the initial space of the system which provides a domain for $X_1(0), \ldots, X_n(0)$, and $\cF$ is a boson Fock space \cite{P_1992} for the action of quantum Wiener processes $W_1(t), \ldots, W_m(t)$. The latter are self-adjoint operators which model the external boson fields. The energetics of the system-field interaction is specified by the system Hamiltonian $H(t)$ and the system-field coupling operators $L_1(t), \ldots, L_m(t)$ which are self-adjoint operators,  representable as time invariant functions (for example, polynomials with constant coefficients) of the system variables $X_1(t), \ldots, X_n(t)$. Both $H(0)$ and $L_1(0), \ldots, L_m(0)$ act on the initial space $\cH$. For what follows, the system and field variables, and the coupling operators are assembled into vectors:
\begin{equation}
\label{XWL}
    X(t):=
    {\begin{bmatrix}
        X_1(t)\\
        \vdots\\
        X_n(t)
    \end{bmatrix}},
    \qquad
    W(t):=
    {\begin{bmatrix}
        W_1(t)\\
        \vdots\\
        W_m(t)
    \end{bmatrix}},
    \qquad
    L(t):=
    {\begin{bmatrix}
        L_1(t)\\
        \vdots\\
        L_m(t)
    \end{bmatrix}}.
\end{equation}
Depending on the context, a system operator $\sigma$, acting on the initial space $\cH$, will sometimes be identified with its extension $\sigma\ox \cI_{\cF}$ to the composite system-field space $\cH\ox \cF$.
The system and the fields form an isolated entity which undergoes a unitary evolution $U(t)$ on $\cH\ox \cF$ driven by the system-field interaction according to the following QSDE \cite{HP_1984,P_1992} with initial condition $U_0:=\cI_{\cH\ox \cF}$:
\begin{equation}
\label{dU}
    \rd U(t) = -\Big(i(H_0\rd t + L_0^{\rT} \rd W(t)) + \frac{1}{2}L_0^{\rT}\Omega L_0\rd t\Big)U(t),
\end{equation}
which corresponds to the Schr\"{o}dinger picture of quantum dynamics.
 Here and in what follows, the initial values of time-varying operators (or vectors or matrices of operators) at $t=0$ are marked by subscript $0$ (so that $H_0:= H(0)$, $L_0:= L(0)$ and $U_0:= U(0)$), while the time dependence will often be omitted for brevity. The QSDE (\ref{dU}) corresponds to a particular yet important case of the identity scattering matrix, when the photon exchange between the fields and the gauge processes \cite{P_1992} can be eliminated from consideration.   Here,
$\Omega:= (\omega_{jk})_{1\<j,k\< m}\in \mH_m^+$ is the quantum Ito matrix of the Wiener process $W$:
\begin{equation}
\label{Omega}
    \rd W\rd W^{\rT} = \Omega \rd t,
    \qquad
    \Omega:= I_m + iJ,
\end{equation}
where the dimension $m$  is assumed to be even, and the matrix $J\in \mA_m$ specifies the cross-commutations between the quantum Wiener processes $W_1, \ldots, W_m$:
\begin{equation}
\label{J}
    [\rd W, \rd W^{\rT}]
    =
    2iJ\rd t,
    \qquad
    J := \bJ\ox I_{m/2},
    \qquad
    \bJ
    :=
    {\begin{bmatrix}
        0& 1\\
        -1 & 0
    \end{bmatrix}}.
\end{equation}
Accordingly, the term $L_0^{\rT}\rd W$ in (\ref{dU})  can be interpreted as an incremental Hamiltonian of the system-field interaction, while   $\frac{1}{2}L_0^{\rT}\Omega L_0\rd t$ is associated with the  Ito correction term in the differential relation $\rd (UU^{\dagger})=  (\rd U)U^{\dagger}+U\rd U^{\dagger} + (\rd U)\rd U^{\dagger}   = 0$ describing the preservation of co-isometry of $U(t)$.
The system variables at time $t\>0$, as operators on the system-field space $\cH\ox \cF$, are the images
\begin{equation}
\label{XU}
  X_k(t) = \sj_t(X_k(0))
\end{equation}
of their initial values
under the flow $\sj_t$ which maps a system operator $\sigma_0$ on the initial space $\cH$ to
\begin{equation}
\label{flow}
    \sigma(t)
    :=
    \sj_t(\sigma_0)
    =
    U(t)^{\dagger}(\sigma_0\ox \cI_{\cF})U(t).
\end{equation}
Note that, in view of (\ref{dU}), the flow $\sj_t$ depends on the energy operators $H_0$ and $L_0$.
The resulting quantum adapted  process $\sigma$ satisfies the following Hudson-Parthasarathy QSDE \cite{HP_1984,P_1992}:
\begin{equation}
\label{dsigma}
    \rd \sigma = \cG(\sigma) \rd t - i[\sigma,L^{\rT}]\rd W,
    \qquad
    \cG(\sigma) := i[H,\sigma] + \cD(\sigma),
\end{equation}
where use is made of the Hamiltonian and the coupling operators evolved by the flow $\sj_t$ from (\ref{flow}) as
\begin{equation}
\label{HL}
    H(t) = \sj_t(H_0),
    \qquad
    L(t) = \sj_t(L_0)
\end{equation}
(the flow  acts entrywise on vectors and matrices of operators).
Also, $\cD$ in (\ref{dsigma}) denotes the Gorini-Kossakowski-Sudar\-shan-Lin\-d\-blad (GKSL) decoherence superoperator \cite{GKS_1976,L_1976} which acts on $\sigma(t)$ as
\begin{align}
\nonumber
    \cD(\sigma)
    & :=
    \frac{1}{2}
    \sum_{j,k=1}^m
    \omega_{jk}
    ( L_j[\sigma,L_k] + [L_j,\sigma]L_k )\\
\nonumber
    & =
    \frac{1}{2}
    \big(
        L^{\rT}\Omega [\sigma,L]
        +
        [L^{\rT},\sigma]\Omega L
    \big)\\
\label{cD}
     & =
    -[\sigma, L^{\rT}]\Omega L - \frac{1}{2} [L^{\rT}\Omega L, \sigma].
\end{align}
The third equality in (\ref{cD}) is convenient for the entrywise application of the superoperator $\cD$ to vectors of operators.
The superoperator $\cG$ in (\ref{dsigma}) is referred to as the GKSL generator.
The identity $(\sigma_0\ox \cI_{\cF})U = U\sigma$, which holds for system operators in view of (\ref{flow}) and
the unitarity of $U(t)$, allows the QSDE (\ref{dU}) to be represented in the Heisenberg picture using (\ref{HL}):
\begin{equation}
\label{dUHL}
    \rd U = -U\Big(i(H\rd t + L^{\rT} \rd W) + \frac{1}{2}L^{\rT}\Omega L\rd t\Big).
\end{equation}
Here, use is also made of the property that for any $t\>0$, the future-pointing increments $\rd W(t)$ commute with adapted processes evaluated at the same (or an earlier) moment of time.
In application to
the vector $X(t)$ of system variables, the flow
\begin{equation}
\label{X}
    X(t):=\sj_t(X_0) = U(t)^{\dagger}(X_0\ox \cI_{\cF})U(t)
\end{equation}
 is understood (also entrywise as mentioned before) in accordance with (\ref{XU}) and (\ref{flow}), and
the corresponding QSDE (\ref{dsigma}) is representable in a vector-matrix form as
\begin{equation}
\label{dX}
    \rd X
    =
    F\rd t +G\rd W,
    \qquad\
    F:= \cG(X),
    \qquad\
    G := -i[X, L^{\rT}],
\end{equation}
where the $n$-dimensional drift vector $F$  and the dispersion $(n\x m)$-matrix $G$ consist of time-varying operators.
Furthermore, the interaction of the input field with the system produces an $m$-dimensional  output field
\begin{equation}
\label{Y}
    Y(t):=
    {\begin{bmatrix}
        Y_1(t)\\
        \vdots\\
        Y_m(t)
    \end{bmatrix}}
    =
    U(t)^{\dagger}(\cI_{\cH}\ox W(t))U(t),
\end{equation}
where the system-field unitary evolution is applied to the current input field variables (which reflects the innovation nature of the quantum Wiener process and the continuous tensor product structure of the Fock space \cite{PS_1972}).  The output field satisfies
the QSDE
\begin{equation}
\label{dY}
    \rd Y = 2JL \rd t + \rd W,
\end{equation}
where the matrix $J$ is given by (\ref{J}), and $L$ is the vector of system-field coupling operators from (\ref{XWL}).
The common unitary evolution in (\ref{X}) and (\ref{Y}) preserves the commutativity between the system variables and the output field variables in time $t\> 0$:
\begin{equation}
\label{XY}
        [X(t),Y(t)^{\rT}] = U(t)^{\dagger}[X_0 \ox \cI_{\cF},\cI_{\cH_0}\ox W(t)^{\rT}]U(t)
     =
    0,
\end{equation}
where the entries of $X_0$ and $W(t)$ commute as operators on different spaces.
By a similar reasoning, the output field $Y$ inherits the CCR matrix $J$ from the input quantum Wiener process $W$:
\begin{align*}
    [Y(t), Y(t)^{\rT}]
    & =
    U(t)^{\dagger}[\cI_{\cH}\ox W(t), \cI_{\cH}\ox W(t)^{\rT}] U(t)\\
    & =
    U(t)^{\dagger}(\cI_{\cH} \ox[W(t), W(t)^{\rT}]) U(t)\\
     & =
    2it J  U(t)^{\dagger}\cI_{\cH\ox \cF} U(t)    = 2itJ.
\end{align*}
This relation can also be established  as a corollary of the property that $\rd Y \rd Y^{\rT} = \rd W\rd W^{\rT} = \Omega \rd t$ in view of (\ref{Omega}), (\ref{J}) and (\ref{dY}).

\section{OPEN QUANTUM HARMONIC OSCILLATOR WITH PARAMETRIC DEPENDENCE}\label{sec:OQHO}

In a particular case, when the system being considered is an OQHO \cite{EB_2005}, its dynamic variables satisfy the Heisenberg CCRs
\begin{equation}
\label{Theta}
    [X(t),X(t)^{\rT}]
    =
    2i
    \Theta
\end{equation}
on a dense domain,
where the CCR matrix $\Theta \in \mA_n$ represents $\Theta \ox \cI_{\cH\ox \cF}$ and remains unchanged. The Hamiltonian of the OQHO is a quadratic polynomial and  the system-field coupling operators in (\ref{XWL}) are linear functions of the system variables:
\begin{equation}
\label{RN}
    H = \frac{1}{2}X^{\rT} R X,
    \qquad
    L
    =
    N X,
\end{equation}
where $R\in \mS_n$ and $N \in \mR^{m\x n}$ are constant matrices. In this case,  (\ref{dX}) and (\ref{dY}) become linear QSDEs
\begin{equation}
\label{dXdY}
    \rd X = A X\rd t + B\rd W,
    \qquad
    \rd Y = C X\rd t + \rd W
\end{equation}
(with linear drift $F=AX$ and dispersion matrix $G=B$ with real-valued entries)
with     constant coefficients which comprise the matrices $A\in \mR^{n\x n}$, $B\in \mR^{n\x m }$ and $C\in \mR^{m\x n}$:
\begin{equation}
\label{ABC}
    A:= 2\Theta (R +N^{\rT}JN),
    \qquad
    B:= 2\Theta N^{\rT},
    \qquad
    C := 2JN.
\end{equation}
These matrices satisfy the following algebraic equations
\begin{equation}
\label{PR}
    A\Theta + \Theta A^{\rT} + BJB^{\rT} = 0,
    \qquad
    \Theta C^{\rT} + BJ = 0
\end{equation}
which describe the fulfillment of the CCRs (\ref{Theta}) and the commutativity (\ref{XY}) at any moment of time and provide necessary and sufficient conditions of physical realizability   \cite{JNP_2008} of the linear  QSDEs (\ref{dXdY}) as an OQHO with the CCR matrix $\Theta$. The first equality  in (\ref{PR}) can be regarded as an algebraic Lyapunov equation (ALE) with respect to $\Theta$ which has a unique solution if and only if the Kronecker sum $A\op A:= I_n \ox A + A\ox I_n$ is a nonsingular matrix (that is, $A$ has no two eigenvalues which are symmetric to each other about the origin of the complex plane).

For general open quantum systems (whose dynamic variables satisfy CCRs, while the Hamiltonian and the coupling operators are not necessarily quadratic and linear functions of the system variables, respectively), the CCR preservation follows from the unitary evolution of the system variables in (\ref{X}).

Using the standard notation  of classical linear control theory \cite{AM_1989,KS_1972},  the quadruple of the state-space realization matrices of the QSDEs (\ref{dXdY}), which specify the input-system-output  map from $W$ to $Y$ for the OQHO, can be written as
\begin{equation}
\label{ABCI}
    {\left[
    \begin{array}{c||c}
        A & B\\
        \hline
        \hline
        C & I_m
    \end{array}
    \right]}.
\end{equation}
If the matrices $R$ and $N$, which specify the Hamiltonian and the coupling operators  in (\ref{RN}), depend smoothly on a scalar parameter $\eps$ (while $W$ is independent of $\eps$),  then so do the matrices $A$, $B$, $C$ in (\ref{ABC}) and the system and output variables which comprise the vectors $X(t)$ and $Y(t)$. The derivatives $X(t)':= \d_{\eps} X(t)$ and $Y(t)':= \d_{\eps} Y(t)$  can be interpreted for any $t\> 0$ as adapted processes with self-adjoint operator-valued entries satisfying the QSDEs
\begin{align}
\label{dX'dY'}
    \rd X' & = (A' X + A X') \rd t + B'\rd W,
    \qquad
    \rd Y' = (C' X + CX')\rd t
\end{align}
(the second of which is, in fact, an ODE)
with zero initial conditions $X_0' = 0$ and $Y_0'=0$ since $X_0$ and $Y_0$ do not depend on  $\eps$. Here,
\begin{equation}
\label{A'B'C'}
    A'  = 2\Theta (R' +N'^{\rT}JN + N^{\rT}JN'),
    \qquad
    B' = 2\Theta N'^{\rT},
    \qquad
    C'  = 2JN'
\end{equation}
are the derivatives  of the matrices $A$, $B$, $C$ from (\ref{ABC}) with respect to $\eps$. In combination with (\ref{dXdY}), the QSDEs (\ref{dX'dY'}) allow the input-system-output map from $W$ to ${\small\begin{bmatrix}Y\\ Y'\end{bmatrix}}$ to be  represented  in a form, similar to  (\ref{ABCI}), as
\begin{equation}
\label{AAA}
    {\left[
    \begin{array}{cc||c}
        A & 0 & B\\
        A' & A & B'\\
        \hline
        \hline
        C & 0 & I_m\\
        C' & C & 0
    \end{array}
    \right]}.
\end{equation}
The special block lower triangular structure of the state dynamics matrix 
in (\ref{AAA}) is closely related to the Gateaux derivative
$$
    (\re^{A})' =
    {\begin{bmatrix}
        0 & I_n
    \end{bmatrix}}
    \exp
    \left(
        {\begin{bmatrix}
        A & 0\\
        A' & A\\
        \end{bmatrix}}
    \right)
    {\begin{bmatrix}
            I_n \\
            0
    \end{bmatrix}}
$$
of the matrix exponential in the direction $A'$; see, for example, \cite{H_2008}. At any time $t\>0$,  both $X(t)'$ and $Y(t)'$ depend linearly on the matrices $R'$ and $N'$ from (\ref{A'B'C'}) in view of the following representation of the solutions of the QSDEs (\ref{dX'dY'}) which regards $A' X\rd t + B'\rd W$ as a forcing term:
\begin{align}
\label{X'}
    X' & = \int_0^t \re^{(t-s)A} (A'X(s)\rd s + B' \rd W(s)),\\
\label{Y'}
    Y' & = \int_0^t (C'X(s) + C X(s)' )\rd s,
\end{align}
where
\begin{equation}
\label{Xlin}
    X(t) = \re^{tA}X_0 + \int_0^t \re^{(t-s)A} B\rd W(s)
\end{equation}
is the unperturbed solution of the first QSDE in (\ref{dXdY}) which does not depend on $R'$ or $N'$. In the case of linear QSDEs, whose coefficients depend smoothly on parameters, the mean square differentiability of their solutions with respect to the parameters can be verified directly by using the closed form (\ref{Xlin}) 
(under certain integrability conditions for the underlying quantum state).

The relations (\ref{A'B'C'})--(\ref{Y'}) exhaust the analysis of sensitivity of the OQHO to the perturbations of the matrices $R$ and $N$ under which the energy operators $H$ and $L$ (\ref{RN}) remain in the corresponding classes of quadratic and linear functions of the system variables. However, this analysis is not applicable to more general perturbations (for example, higher order polynomials of the system variables) and is restricted to linear systems.


\section{TRANSVERSE HAMILTONIAN}\label{sec:TH}

We will now consider the response of the general open quantum system, described in Section~\ref{sec:QSDE} and governed by (\ref{dX}) and (\ref{dY}), to perturbation of the system Hamiltonian and the system-field coupling operators:
\begin{equation}
\label{KM}
    H_0\mapsto H_0 + \eps K_0,
    \qquad
    L_0\mapsto L_0 + \eps M_0.
\end{equation}
Here,  $K_0$ and the entries of the $m$-dimensional vector $M_0$ are self-adjoint operators on the initial system space $\cH$, and $\eps$ is a small real-valued  parameter, so that $K_0=H_0'$ and $M_0=L_0'$.  In what follows, the derivative $(\cdot)':= \d_{\eps}(\cdot)$ will be taken at $\eps=0$.  The   perturbations $K_0$ and $M_0$ in (\ref{KM}) are also assumed to be functions of the system variables which are evolved by the \emph{unperturbed} flow (\ref{flow}) as
\begin{equation}
\label{KMflow}
    K(t):= \sj_t(K_0),
    \qquad
    M(t):= \sj_t(M_0).
\end{equation}
For example, in the case of the OQHO considered in Section~\ref{sec:OQHO}, the perturbations $K$ and $M$ inherit the structure of the Hamiltonian and the coupling operators in (\ref{RN}) as a quadratic function and a linear function of the system variables, respectively:
\begin{equation}
\label{KMlin}
    K = \frac{1}{2}X^{\rT} R' X,
    \qquad
    M
    =
    N' X.
\end{equation}
Returning to the general case, we will avoid at this stage technical assumptions on $K_0$ and $M_0$ in (\ref{KM}), so that the calculations, carried out below for arbitrary perturbations, should be regarded as formal. Since the energy operators $H_0$ and $L_0$ completely specify the dynamics of the unitary operator $U(t)$ which determines the evolution of the system and output field variables, the response of the latter to the perturbation of $H_0$ and $L_0$ reduces to that of $U(t)$. Therefore, the propagation of the initial perturbations of the energy operators through the subsequent unitary system-field evolution can be described in terms of the operator
\begin{equation}
\label{V}
    V(t):= U(t)',
\end{equation}
which satisfies $V_0 = 0$ since $U_0=\cI_{\cH\ox \cF}$ does not depend on $\eps$. The smoothness of dependence of $U(t)$ on the parameter $\eps$ is analogous to the corresponding property of solutions of classical SDEs (whose drift and dispersion satisfy suitable regularity conditions \cite{KS_1991,S_2008}) and holds at least in the linear case discussed in Section~\ref{sec:OQHO}.

\begin{theorem}
\label{th:Q}
For any time $t\> 0$, the operator $V(t)$ in (\ref{V}), associated with the unitary evolution $U(t)$ from (\ref{dU}), is representable as
\begin{equation}
\label{VUQ}
    V(t) = -iU(t)Q(t).
\end{equation}
Here, $Q(t)$ is a self-adjoint operator on $\cH\ox \cF$ which satisfies zero initial condition $Q_0 = 0$ and is governed by the QSDE
\begin{equation}
\label{dQ}
    \rd Q
    =
        \big(K -\Im (L^{\rT}\Omega M)\big)\rd t
        +
        M^{\rT}\rd W.
\end{equation}
Furthermore, $Q(t)$  depends linearly on the initial perturbations $K_0$ and $M_0$ of the energy operators in (\ref{KM}) through their unperturbed evolutions in (\ref{KMflow}). 
\end{theorem}
\begin{proof}
The differentiation of both sides of the unitarity relation $U(t)^{\dagger}U(t) = \cI_{\cH\ox \cF}$ with respect to the parameter $\eps$ and letting $\eps=0$  leads to the equalities
$
    V^{\dagger}U + U^{\dagger}V
    =
    (U^{\dagger}V)^{\dagger} + U^{\dagger}V
    =
    0
$,
the second of which implies self-adjointness of the operator
\begin{equation}
\label{Q}
    Q(t):= i U(t)^{\dagger}V(t),
\end{equation}
thus establishing (\ref{VUQ}).
Now, consider the time evolution of $Q(t)$. To this end, the differentiation of (\ref{dU}) with respect to $\eps$ yields
\begin{equation}
\label{dV}
    \rd V =
    -\big(i(K_0\rd t + M_0^{\rT} \rd W)+ \Re(L_0^{\rT}\Omega M_0)\rd t\big)U
    -\Big(i(H_0\rd t + L_0^{\rT} \rd W) + \frac{1}{2}L_0^{\rT}\Omega L_0\rd t\Big)V.
\end{equation}
By left multiplying both sides of (\ref{dV}) by $U^{\dagger}$ and recalling (\ref{Q}), it follows that
\begin{equation}
 \label{UdV}
    U^{\dagger}\rd V =
    -i(K\rd t + M^{\rT} \rd W)- \Re(L^{\rT}\Omega M)\rd t
+\Big(\frac{i}{2}L^{\rT}\Omega L\rd t-H\rd t - L^{\rT} \rd W\Big) Q,
\end{equation}
where use is made of the evolved perturbations (\ref{KMflow}).
By a similar reasoning, a combination of (\ref{dUHL}) with (\ref{Q}) leads to
\begin{align}
    (\rd U^{\dagger}) V
    & =
    \Big(i\big(H\rd t + L^{\rT} \rd W\big) - \frac{1}{2}L^{\rT}\Omega L\rd t\Big)U^{\dagger}V
     =
\Big(H\rd t + L^{\rT} \rd W + \frac{i}{2}L^{\rT}\Omega L\rd t\Big) Q,\\
\label{dUdV}
    \rd U^{\dagger} \rd V
    &=
    iL^{\rT} \rd W U^{\dagger}\rd V
    =
    -iL^{\rT} \rd W\rd W^{\rT} (iM+ LQ)
     =
    L^{\rT}\Omega (M - iLQ)\rd t,
\end{align}
where use is also made of (\ref{UdV}) and the quantum Ito product rules \cite{P_1992} in combination with (\ref{Omega}). It now follows from (\ref{Q}) and (\ref{UdV})--(\ref{dUdV}) that
\begin{align*}
    \rd Q & =
        i((\rd U^{\dagger}) V + U^{\dagger} \rd V + \rd U^{\dagger} \rd V)\\
        & =
        \big(
            K + iL^{\rT}\Omega M- i\Re(L^{\rT}\Omega M)
        \big)
        \rd t
        +
        M^{\rT}\rd W\\
     & =
        \big(
            K -\Im (L^{\rT}\Omega M)
        \big)
        \rd t
        +
        M^{\rT}\rd W,
\end{align*}
which establishes (\ref{dQ}). The linear dependence of $Q(t)$ on $K_0$ and $M_0$ follows from the integral representation
\begin{equation}
\label{Qint}
    Q(t)
    =
    \int_0^t
        \big(
            (K(s) -\Im (L(s)^{\rT}\Omega M(s)))\rd s
            +
            M(s)^{\rT}\rd W(s)
        \big)
\end{equation}
of the QSDE (\ref{dQ})
and the property that the evolved perturbations $K$ and $M$ in (\ref{KMflow}) depend linearly on $K_0$ and $M_0$, respectively.
\end{proof}

In view of Theorem~\ref{th:Q}, the equation $U(t)'=-iU(t)Q(t)$ has the structure of isolated quantum dynamics in fictitious time $\eps$, where $Q(t)$ plays the role of a Hamiltonian  associated with the perturbation of the unitary operator $U(t)$. In order to reflect this property, the time-varying operator $Q$ will be referred to as the \emph{transverse Hamiltonian} associated with the perturbations $K$ and $M$ of the energy operators.  Its computation is illustrated below.

\begin{example}
In the absence of perturbation to the system-field coupling, when $M_0$ consist of zero operators, the transverse Hamiltonian in  (\ref{Qint}) reduces to $Q(t) = \int_0^t K(s)\rd s$.  Moreover, if the system is isolated (and hence, the unitary process $U$ in (\ref{dU}) takes the form $U(t) = \re^{-itH_0}$),  then
$$
    Q(t)
    =
    \int_0^t \re^{isH_0}K_0 \re^{-isH_0}\rd s
    = \int_0^t \re^{is \ad_{H_0}}(K_0)\rd s
    =
    tE(it\ad_{H_0})(K_0).
$$
Here, use is made of Hadamard's lemma \cite{M_1998}, and $E(z):= \frac{\re^z-1}{z} = \sum_{k=0}^{+\infty} \frac{z^k}{(k+1)!}$ is an entire function which plays a part 
in the solution of systems of nonhomogeneous linear ODEs with constant coefficients \cite
{H_2008}. \hfill$\blacktriangle$
\end{example}

\begin{example}
For the OQHO of Section~\ref{sec:OQHO}, substitution of the perturbations (\ref{KMlin}) into (\ref{Qint}) leads to the transverse Hamiltonian
\begin{align}
\nonumber
    Q(t)
    =&
    \int_0^t
        \Big(
            \frac{1}{2}
            X(s)^{\rT}
            \big(
                R'
                +i
                (N^{\rT}\Omega N' - N'^{\rT}\Omega N)
            \big)
            X(s)\rd s
             +
            X(s)^{\rT}N'^{\rT}\rd W(s)
        \Big)\\
    \label{Qlin}
        =&
    \int_0^t
        \Big(
            \frac{1}{2}
            X(s)^{\rT}
            \big(
                R'
                +
                N'^{\rT}J N
                -N^{\rT}J N'
            \big)
            X(s)\rd s
             +
            X(s)^{\rT}N'^{\rT}\rd W(s)        \Big)
            -
            \bra
                N\Theta, N'
            \ket\, t
\end{align}
which depends in a quadratic fashion on the past history of the system variables. The latter are given by the unperturbed equation (\ref{Xlin}). The last term $\bra
                N\Theta, N'
            \ket\, t$ in (\ref{Qlin}) comes  from the relation $\Im (i (N^{\rT}\Omega N' - N'^{\rT}\Omega N)) = N^{\rT}N'-N'^{\rT}N$ (following  from (\ref{Omega})) and the identity $X^{\rT} \Ups X = i\bra \Ups, \Theta\ket$  which holds for any $\Ups \in \mA_n$ in view of the CCRs (\ref{Theta}).\hfill$\blacktriangle$
\end{example}

\section{PERTURBATION ANALYSIS OF SYSTEM OPERATORS}\label{sec:sys}

Since the transverse Hamiltonian $Q(t)$ encodes the propagation of the initial perturbations of the energy operators in (\ref{KM}) through the unitary system-field evolution over the time interval $[0,t]$, it provides a tool for the infinitesimal perturbation analysis of general system operators. The following theorem is concerned with an extended setting which, in addition to (\ref{KM}), allows for an initial perturbation in a system operator $\sigma_0$ specified by $\sigma_0'$.

\begin{theorem}
\label{th:dashQSDE}
For any self-adjoint system operator $\sigma_0$ on the initial space $\cH$, which smoothly depends on the same scalar parameter $\eps$ as in (\ref{KM}) and is evolved by the flow (\ref{flow}), the derivative of its evolved version $\sigma(t)$ with respect to $\eps$  is representable as
\begin{equation}
\label{flow'}
    \tau(t)
    :=
    \sigma(t)'
    =
    \sj_t(\sigma_0')
    +
    \phi(t),
    \qquad
    \phi(t)
    :=
    i[Q(t), \sigma(t)]
\end{equation}
at any time $t\> 0$,
where $Q(t)$ is   the transverse Hamiltonian  from Theorem~\ref{th:Q}. Here, the operator $\phi(t)$ satisfies the QSDE
\begin{equation}
\label{dphi}
    \rd \phi
    =
    (i[Q,\cG(\sigma)] + \chi(\sigma))\rd t
     + \big([Q,[\sigma,L^{\rT}]]-i[\sigma,M^{\rT}]\big)\rd W,
\end{equation}
with zero initial condition $\phi_0=0$, where $\cG$ is the GKSL generator from (\ref{dsigma}), and
$\chi$ is a linear superoperator given by
\begin{equation}
\label{chi}
    \chi(\sigma)
    :=
        i[K -\Im (L^{\rT}\Omega M),\sigma]
        -2\Re ([\sigma,L^{\rT}]\Omega M).
\end{equation}
\end{theorem}

\begin{proof}
By using the Leibniz product rule together with (\ref{flow}), (\ref{V}), (\ref{VUQ}) and the self-adjointness of $Q(t)$, it follows that
\begin{align*}
    \sigma' & =
    U^{\dagger}(\sigma_0'\ox \cI_{\cF})U
    +
    V^{\dagger}(\sigma_0\ox \cI_{\cF})U + U^{\dagger}(\sigma_0\ox \cI_{\cF})V\\
     & = \sj_t(\sigma_0')
     +(-iUQ)^{\dagger}(\sigma_0\ox \cI_{\cF})U- iU^{\dagger}(\sigma_0\ox \cI_{\cF})UQ\\
     & =
     \sj_t(\sigma_0')
     +
    iQU^{\dagger}(\sigma_0\ox \cI_{\cF})U- i\sigma Q
      = \sj_t(\sigma_0') + i[Q, \sigma],
\end{align*}
which establishes (\ref{flow'}). The equality $\varphi_0=0$ follows from $Q_0=0$.
We will now obtain a QSDE for the process $\phi$. To this end, by combining the quantum Ito lemma \cite{HP_1984,P_1992} with the bilinearity of the commutator, and using the QSDEs (\ref{dsigma}) and (\ref{dQ}), it follows that
\begin{align}
\nonumber
    \rd [Q,\sigma]
    &=
    [\rd Q, \sigma] + [Q, \rd \sigma] + [\rd Q, \rd \sigma]\\
\nonumber
    &=
     [K -\Im (L^{\rT}\Omega M),\sigma]\rd t -[\sigma, M^{\rT}]\rd W+ [Q, \cG(\sigma)]\rd t- i[Q,[\sigma,L^{\rT}]]\rd W -[M^{\rT}\rd W, i[\sigma,L^{\rT}]\rd W]\\
\nonumber
    &=
    \big([K -\Im (L^{\rT}\Omega M),\sigma]
    + [Q, \cG(\sigma)] +2i \Im (i[\sigma,L^{\rT}]\Omega M)
          \big)\rd t - \big([\sigma,M^{\rT}] + i[Q,[\sigma,L^{\rT}]]\big)\rd W\\
\label{dQsig}
    &=
    \big(
            [K -\Im (L^{\rT}\Omega M),\sigma]
    + [Q, \cG(\sigma)] +2i \Re ([\sigma,L^{\rT}]\Omega M)
    \big)\rd t
     - \big([\sigma,M^{\rT}]+i[Q,[\sigma,L^{\rT}]]\big)\rd W.
\end{align}
Here, the quantum Ito product rules are combined with the commutativity between adapted processes and future-pointing increments of the quantum Wiener process. In the second and third of the equalities (\ref{dQsig}), use is also made of the relations $[\alpha,\beta^{\rT}\rd W] =[\alpha,\beta^{\rT}]\rd W$ and  $[\alpha^{\rT}\rd W, \beta^{\rT}\rd W] = 2i\Im(\alpha^{\rT} \Omega \beta)\rd t$ for appropriately dimensioned
adapted processes $\alpha$ and $\beta$ with self-adjoint operator-valued entries. These are combined with the identity $ \Im (i[\sigma,L^{\rT}]\Omega M) = \Re ([\sigma,L^{\rT}]\Omega M)$ in the fourth equality of (\ref{dQsig}). The QSDE (\ref{dphi}) can now be obtained my multiplying both sides of (\ref{dQsig}) by $i$ and using (\ref{chi}).
\end{proof}

The process $\sj_t(\sigma_0')$ in (\ref{flow'}) is the response of $\sigma(t)$ to the initial perturbation $\sigma_0'$ of the system operator under the unperturbed flow (\ref{flow}) and  satisfies the  QSDE (\ref{dsigma}):
\begin{equation}
\label{djsig}
    \rd \sj_t(\sigma_0') = \cG(\sj_t(\sigma_0')) \rd t - i[\sj_t(\sigma_0'),L^{\rT}]\rd W.
\end{equation}
The operator $\phi(t)$ in (\ref{flow'}) describes the response of the flow $\sj_t$ itself to the perturbation of the energy operators and will be referred to as the \emph{derivative  process} for the system operator $\sigma$. Accordingly, the term $i[Q,\cG(\sigma)]$ in the drift of the QSDE (\ref{dphi}) is the derivative process for $\cG(\sigma)$.
The last equality in (\ref{flow'}) is organized as the right-hand side of the Heisenberg ODE (in fictitious time $\eps$) of an isolated quantum system with the state space $\cH\ox \cF$ and the transverse Hamiltonian $Q(t)$. While the unperturbed flow $\sj_t(\sigma_0')$ depends linearly on $\sigma_0'$, the derivative process $\phi(t)$ depends linearly on the perturbations $K_0$ and $M_0$ of the energy operators through the transverse Hamiltonian $Q$ and the superoperator $\chi$ in (\ref{chi}).

In application to the vectors $X$ and $L$ of the system variables and system-field coupling operators in (\ref{XWL}), the QSDEs (\ref{dphi}), (\ref{djsig}) and the definition (\ref{chi}) lead to
\begin{align}
\nonumber
    \rd X'
    =&
    \big(
        i(
        [K -\Im (L^{\rT}\Omega M),X]
        + [Q, F]
        )
        +2\Im (G\Omega M)
    \big)\rd t\\
\label{X'QSDE}
     & + i([Q,G]-[X,M^{\rT}]\big)\rd W,\\
\nonumber
    \rd L'=&
    \big(
    \cG(M)+
        i(
        [K -\Im (L^{\rT}\Omega M),L]
        + [Q, \cG(L)]
        ) -2\Re ([L,L^{\rT}]\Omega M)
    \big)\rd t\\
\label{L'QSDE}
    & + \big([Q,[L,L^{\rT}]]-i[L,M^{\rT}] - i[M,L^{\rT}]\big)\rd W,
\end{align}
where $F$ and $G$ are the drift vector and the dispersion matrix from (\ref{dX}). In (\ref{X'QSDE}), use is also made  of the  absence $X_0' = 0$ of initial perturbations in the system variables, whereby  $\sj_t(X_0') = 0$ for any $t\>0$. Furthermore,
in (\ref{L'QSDE}), we have used the relations $\sj_t(L_0') = \sj_t(M_0) = M(t)$ in view of (\ref{KM}), (\ref{KMflow}). Since the matrix $J$ and the quantum Wiener process $W$ do not depend on the parameter $\eps$, the derivative of the output field $Y$ of the system  in (\ref{dY}) evolves according to the ODE
\begin{equation}
\label{Y'QSDE}
    (Y')^{^\centerdot} = 2JL',
\end{equation}
where $(\ )^{^\centerdot}$ denotes the time derivative, and $L'$ is governed by the QSDE (\ref{L'QSDE}).

In particular, for the OQHO of Section~\ref{sec:OQHO}
with  the perturbations (\ref{KMlin}) of the energy operators,
the QSDEs  (\ref{X'QSDE})--(\ref{Y'QSDE}) 
lead to the relations (\ref{dX'dY'}), (\ref{A'B'C'}) which were obtained in Section~\ref{sec:OQHO}  using more elementary techniques. However, the latter are limited to the class of linear-quadratic  perturbations (\ref{KMlin}) which inherit the structure of the corresponding energy operators, whereas the transverse Hamiltonian approach allows the system response to be investigated for general perturbations of the energy operators. Therefore, this approach can be used for the development of optimality conditions in coherent quantum control problems for larger classes of controllers and filters.

\section{SENSITIVITY OF QUANTUM AVERAGED PERFORMANCE FUNCTIONALS}\label{sec:func}

Following the classical control paradigm for stochastic systems \cite{AM_1989,KS_1972}, suppose the desired operation of the quantum system being considered over an infinite time horizon  corresponds to small values of an averaged cost functional
\begin{equation}
\label{cZ}
    \cZ
    :=
    \lim_{t\to +\infty}
        \bE Z(t),
\end{equation}
provided the limit exists. The latter assumption guarantees the same Cesaro limit $    \lim_{t\to +\infty}
    \big(
        \frac{1}{t}
        \int_0^t
        \bE Z(s)
        \rd s
    \big) = \cZ
$.
The expectation $\bE(\cdot)$ is taken over the tensor product $\rho := \varpi \ox \upsilon$ of the initial quantum state $\varpi$ of the system and the vacuum state $\upsilon$ in the boson Fock space $\cF$ for the external fields. Also, it is assumed that $Z$ (which will be referred to as the \emph{criterion process}) is an adapted quantum  process given by
\begin{equation}
\label{Z}
    Z(t) := f(X(t), L(t)).
\end{equation}
Here, $f: \mR^{n+m}  \to \mR$ is a function which is extended to the noncommutative system variables and the ``signal'' part $2JL$ of the  output field in (\ref{dY}) so as to make $Z(t)$ a self-adjoint operator for any $t\> 0$ (note that $\det J\ne 0$ in view of (\ref{J})). Such extension is straightforward in the case of polynomials and can be carried out through the Weyl quantization \cite{F_1989} for more general functions. In the CQLQG control and filtering problems \cite{NJP_2009,VP_2013a,VP_2013b}, the function $f$ in (\ref{Z}) is a positive semi-definite quadratic form.
Now, if the energy operators of the quantum system are subjected to perturbations (\ref{KM}), then by using the transverse Hamiltonian $Q$ from Theorem~\ref{th:Q} in combination with Theorem~\ref{th:dashQSDE}, it follows that
\begin{equation}
\label{phiZ}
    (\bE Z(t))' 
    = \bE(\sj_t(Z_0') + \phi(t)),
    \qquad
    \phi(t):= i[Q(t),Z(t)].
\end{equation}
Here, $Z_0'$ can be nonzero due to the dependence on $L$ in (\ref{Z}) despite $X_0'=0$, and $\phi$ is the derivative process for $Z$.
Ignoring, at this stage, the technical issues of existence and interchangeability of appropriate limits,
(\ref{phiZ}) leads to the following perturbation formula for the cost functional $\cZ$ in (\ref{cZ}):
\begin{equation}
\label{cZ'}
    \cZ'
    :=
    \lim_{t\to+\infty}
        (
        \bE\sj_t(Z_0')
        +
        \bE \phi(t)
        ).
\end{equation}
The right-hand side of (\ref{cZ'}) is a linear functional of the perturbations $K_0$ and $M_0$ which describes the corresponding (formal) Gateaux derivative  of $\cZ$. Therefore, the quantum system is a stationary point  of the performance criterion (\ref{cZ}) with respect to a subspace  $\cT$  of perturbations $(K_0, M_0)$ of the energy operators in (\ref{KM}), if $\cT$ is contained by the null space of the linear functional $\cZ'$ in (\ref{cZ'}):
\begin{equation}
\label{stat}
    \cT\subset \ker \cZ'.
\end{equation}
While $\lim_{t\to+\infty}\bE\sj_t(Z_0') $ in (\ref{cZ'}) reduces to averaging over the invariant state of the unperturbed system (provided the latter is ergodic and certain integrability conditions are satisfied), the computation of $\lim_{t\to+\infty} \bE \phi(t)$ is less straightforward.
Due to the product structure of the system-field state $\rho$ (with the external fields  being in the vacuum state),  the martingale part of the QSDE (\ref{dphi}) does not contribute to the time derivative
\begin{equation}
\label{Ephidot}
    (\bE \phi)^{^\centerdot} = \bE\psi  + \bE \chi(Z),
    \qquad
    \psi := i[Q,\cG(Z)],
\end{equation}
where $\psi$ is the derivative process for $\cG(Z)$, and the superoperator  $\chi$ from (\ref{chi}) is applied to the criterion process $Z$. The relation (\ref{Ephidot}) is a complicated integro-differential equation (IDE). However, this IDE admits an efficient solution, for example,  in the case when the system is an OQHO, and the function $f$ in (\ref{Z}) is a polynomial. In this case, due to the structure of the GKSL generator of the OQHO, $\cG(Z)$ is also a polynomial of the system variables of the same degree, thus leading to a linear relation (with constant coefficients) between the derivative processes $\phi$ and $\psi$, and algebraic closedness in the IDE (\ref{Ephidot}). Therefore, for OQHOs with a Hurwitz matrix $A$ and a polynomial criterion process $Z$, the computation  of $\cZ'$ (and verification of the stationarity condition (\ref{stat})) reduces to averaging over the unperturbed invariant state (which is unique and Gaussian).

\begin{example}\label{ex:OQHOquad}
Consider the OQHO of Section~\ref{sec:OQHO}, described by (\ref{Theta})--(\ref{ABC}) with a Hurwitz matrix $A$.
Suppose the criterion process $Z$ in (\ref{Z}) is a quadratic form:
\begin{equation}
\label{Zex}
    Z :=
    \frac{1}{2}
    {\begin{bmatrix}
        X^{\rT} & L^{\rT}
    \end{bmatrix}}
    \Pi
    {\begin{bmatrix}
        X\\
        L
    \end{bmatrix}}
    =
    \frac{1}{2}
    X^{\rT}
    P
    X,
    \qquad
\Pi:=
    {\begin{bmatrix}
        \Pi_{11} & \Pi_{12}\\
        \Pi_{21} & \Pi_{22}
    \end{bmatrix}},
\end{equation}
where the $\frac{1}{2}$ factor is introduced  for further convenience, $\Pi
 \in \mS_{n+m}^+$ is a given appropriately partitioned weighting matrix, and $P\in \mS_n^+$ is an auxiliary matrix:
 \begin{equation}
\label{P}
  P:= \Pi_{11} + \Pi_{12} N + N^{\rT} \Pi_{21} + N^{\rT} \Pi_{22} N.
\end{equation}
 Then, in view of $X_0'=0$, it follows from (\ref{Zex}) and (\ref{KMflow}) that
\begin{equation}
\label{Z0'ex}
    \sj_t(Z_0') = \Re ((\Pi_{21} X + \Pi_{22} L)^{\rT} M) = \Re (X^{\rT}(\Pi_{12}  + N^{\rT}\Pi_{22}) M).
\end{equation}
The second equality in (\ref{Zex}) allows the quantum average of the corresponding derivative process $\phi$ in (\ref{phiZ}) to be represented as
\begin{equation}
\label{XiUps}
    \bE \phi = \frac{1}{2}\bra P, \Ups\ket,
    \qquad
    \Ups:= i \bE [Q, \Xi],
    \qquad
    \Xi:= \Re(XX^{\rT}) = XX^{\rT} -i\Theta.
\end{equation}
Here, $\Ups$ is the expectation of the derivative process $i[Q,\Xi]$ for $\Xi$ and hence, satisfies the following  IDE,
similar to (\ref{Ephidot}):
\begin{equation}
\label{EUpsdot}
    \dot{\Ups} = i \bE [Q,\cG(\Xi)] + \bE \chi(\Xi),
\end{equation}
where the superoperator $\chi$ from (\ref{chi}) is applied to $\Xi$ entrywise as
\begin{align}
\nonumber
    \chi(X_jX_k)
    &=
    i[K -\Im (L^{\rT}\Omega M),X_jX_k]
        -
        2\Re ([X_jX_k,L^{\rT}]\Omega M)    \\
\label{chiXX}
    &=
    i[K -\Im (X^{\rT}N^{\rT}\Omega M),X_jX_k]
         +4\Im ((X_k \Theta_{j\bullet} + X_j \Theta_{k\bullet}) N^{\rT}\Omega M),
\end{align}
with $\Theta_{\ell \bullet}$ the $\ell$th row of the CCR matrix $\Theta$. By substituting the Hamiltonian and coupling operators of the OQHO from (\ref{RN}) into the GKSL generator $\cG$ in (\ref{dsigma}) and (\ref{cD}) and using the CCRs (\ref{Theta}) together with the state-space matrices (\ref{ABC}), it follows that 
\begin{equation}
\label{cGXX}
    \cG(\Xi)
     =
    A\Xi + \Xi A^{\rT} + BB^{\rT}.
\end{equation}
Substitution of (\ref{cGXX}) into (\ref{EUpsdot}) reduces the IDE to a Lyapunov ODE:
\begin{equation}
\label{UpsLODE}
    \dot{\Ups}
     = i \bE [Q,A\Xi + \Xi A^{\rT} + BB^{\rT}] + \bE \chi(\Xi)
    =
    A\Ups + \Ups A^{\rT}+ \bE \chi(\Xi).
\end{equation}
Since the matrix $A$ is Hurwitz, the unperturbed OQHO is ergodic and its invariant state is Gaussian with zero mean $\bE X = 0$ and covariance matrix $\bE(XX^{\rT}) = \Sigma + i\Theta$, where $\Sigma$ is a unique solution of the ALE
$
    A\Sigma + \Sigma A^{\rT} + BB^{\rT} = 0
$.
Therefore, under integrability conditions for $\chi(\Xi)$, to be discussed elsewhere, the solution  of the Lyapunov ODE (\ref{UpsLODE}) has a limit $\Ups_{\infty}:= \lim_{t\to +\infty} \Ups(t)$ which is a unique solution of the ALE
\begin{equation}
\label{UpsALE}
    A\Ups_{\infty} + \Ups_{\infty} A^{\rT}+ \lim_{t\to +\infty}\bE \chi(\Xi) = 0,
\end{equation}
where $\lim_{t\to +\infty}\bE \chi(\Xi)$ can be computed by averaging $\chi(\Xi)$ over the invariant Gaussian state of the OQHO. Therefore, by assembling (\ref{Z0'ex}) and (\ref{XiUps}) into (\ref{cZ'}), it follows that
\begin{equation}
\label{cZ'ex}
    \cZ'
    =
    \lim_{t\to +\infty}
    \bE
    \Re (X^{\rT}(\Pi_{12}  + N^{\rT}\Pi_{22}) M)
    +
    \frac{1}{2}\bra P,\Ups_{\infty}\ket,
\end{equation}
where the limit also reduces to averaging over the invariant Gaussian state. The relations (\ref{P}), (\ref{chiXX}), (\ref{UpsALE}) and (\ref{cZ'ex}) express the Gateaux derivative $\cZ'$ of the quadratic cost functional (\ref{cZ}), specified by (\ref{Zex}), in terms of mixed moments of the system variables $X$ and the perturbations $K$ and $M$ over the invariant Gaussian quantum state. In particular, if $K$ and $M$ are polynomials of the system variables, such moments can be computed using the Isserlis-Wick  theorem \cite{J_1997}. If $K$ and $M$ are trigonometric polynomials, that is, linear combinations of unitary Weyl operators $\re^{iu^{\rT} X}$, with $u\in \mR^n$, the moments can also be computed using quantum Price's theorem \cite{V_2014b}, which will be discussed elsewhere.
\hfill$\blacktriangle$

%
%

\end{example}

\section{MEAN SQUARE OPTIMAL COHERENT QUANTUM FILTERING PROBLEM}\label{sec:CQF}

We will now outline an application of the transverse Hamiltonian technique of Sections~\ref{sec:TH}--\ref{sec:func} to a quantum filtering problem.
Let the open quantum system, described in Section~\ref{sec:QSDE}, be interpreted as a quantum plant. Recall that it has
the Hamiltonian $H$, system-field coupling $L$, input field  $W$, dynamic variables $X$  and output field $Y$ governed by the QSDEs (\ref{dX}) and (\ref{dY}). Suppose the plant is cascaded with another open quantum system, further referred to as the filter,  which is driven by $Y$ and a $\mu$-dimensional quantum Wiener process $\omega$ on a boson Fock space $\fF$; see Fig.~\ref{fig:filtering}.
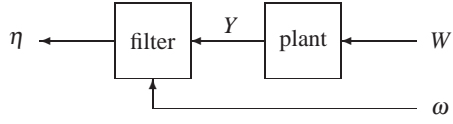
\begin{figure}[htbp]
\centering
\unitlength=1mm
\linethickness{0.2pt}
\begin{picture}(50.00,30.00)
    \put(10,20){\framebox(10,10)[cc]{{\small filter}}}
    \put(30,20){\framebox(10,10)[cc]{{\small plant}}}
    \put(15,16){\vector(0,1){4}}
    \put(15,16){\line(1,0){35}}
 \put(52,16){\makebox(0,0)[lc]{{\small$\omega$}}}
    \put(10,25){\vector(-1,0){10}}
    \put(50,25){\vector(-1,0){10}}
    \put(30,25){\vector(-1,0){10}}
    \put(-2,25){\makebox(0,0)[rc]{{\small$\eta$}}}
    \put(52,25){\makebox(0,0)[lc]{{\small$W$}}}
    \put(25.5,26){\makebox(0,0)[cb]{{\small$Y$}}}
\end{picture}\vskip-15mm
\caption{The series connection of a quantum plant and a quantum filter, mediated by the field $Y$ and affected by the environment through the quantum Wiener processes $W$, $\omega$.
The signal part of the filter output $\eta$ is to approximate the plant variables of interest in the mean square sense.
}
\label{fig:filtering}
\end{figure}
The filter is endowed with its own initial Hilbert space $\fH$, dynamic variables  $\xi_1(t), \ldots, \xi_{\nu}(t)$ with a CCR matrix $\Lambda \in \mA_{\nu}$, and an $(m+\mu)$-dimensional output field $\eta(t)$:
\begin{equation}
\label{xiomlam}
    \xi :=
    {\begin{bmatrix}
        \xi_1 \\
        \vdots\\
        \xi_{\nu}
    \end{bmatrix}},
    \qquad
    \omega :=
    {\begin{bmatrix}
        \omega_1 \\
        \vdots\\
        \omega_{\mu}
    \end{bmatrix}},
    \qquad
    \eta :=
    {\begin{bmatrix}
        \eta_1 \\
        \vdots\\
        \eta_{m+\mu}
    \end{bmatrix}}.
\end{equation}
We denote the filter Hamiltonian by $\Gamma$ and the vectors of operators of coupling of the filter with the plant output $Y$ and the quantum Wiener process $\omega$ by
\begin{equation}
\label{PhiPsi}
    \Phi :=
    {\begin{bmatrix}
        \Phi_1 \\
        \vdots\\
        \Phi_m
    \end{bmatrix}},
    \qquad
    \Psi :=
    {\begin{bmatrix}
        \Psi_1 \\
        \vdots\\
        \Psi_{\mu}
    \end{bmatrix}},
\end{equation}
respectively.
The Hamiltonian $\Gamma$ and the coupling operators $\Phi_1, \ldots, \Phi_m$ and $\Psi_1, \ldots, \Psi_{\mu}$ are functions of the dynamic variables $\xi_1, \ldots, \xi_{\nu}$ of the filter and hence, commute with the plant variables and functions thereof, including $H$ and $L$.
The plant and filter form a composite  open quantum stochastic system, whose vector $\cX$ of dynamic variables satisfies the CCRs
\begin{equation}
\label{cX}
  [\cX, \cX^{\rT}] = 2i\bTheta,
  \qquad
    \bTheta
    :=
    {\begin{bmatrix}
    \Theta & 0\\
    0 & \Lambda
    \end{bmatrix}},
    \qquad
    \cX
    :=
    {\begin{bmatrix}
        X\\
        \xi
    \end{bmatrix}}
\end{equation}
and is driven by an augmented quantum Wiener process $\cW$ with the Ito table
\begin{equation}
\label{cW}
    \rd \cW\rd \cW^{\rT}
    =
    \bOmega \rd t,
    \qquad
    \bOmega:=
    {\begin{bmatrix}\Omega & 0\\
    0 & \mho
    \end{bmatrix}},
    \qquad
    \cW
    :=
    {\begin{bmatrix}
        W\\
        \omega
    \end{bmatrix}}.
\end{equation}
Here, $\mho$ denotes the Ito matrix of the quantum Wiener process $\omega$  of the filter, which is defined similarly to $\Omega$ from (\ref{Omega}) and (\ref{J}) as
\begin{equation}
\label{Lambda}
    \rd \omega\rd \omega^{\rT} = \mho \rd t,
    \qquad
    \mho:= I_{\mu} + i\bJ\ox I_{\mu/2},
\end{equation}
where the dimension $\mu$ is also assumed to be even.
The quantum feedback  network formalism \cite{GJ_2009} allows the Hamiltonian $\bH$ of the plant-filter system and the vector $\bL$ of operators of coupling with $\cW$ to be computed as
\begin{align}
\label{bHbL}
    \bH = H + \Gamma + \Phi^{\rT} J L,
    \qquad
    \bL= {\begin{bmatrix}L + \Phi\\ \Psi\end{bmatrix}}.
\end{align}
%
%
%
%
%
%
Hence, the filter variables are governed by the QSDE
\begin{equation}
\label{dxi}
    \rd \xi
     =
    \bG(\xi)\rd t
    -
    i[\xi,\bL^{\rT}] \rd \cW
     =
    (i[\Gamma,\xi] + \Delta(\xi))\rd t
    -
    i\big[\xi, {\begin{bmatrix}\Phi^{\rT} & \Psi^{\rT}\end{bmatrix}}\big]\rd {\begin{bmatrix}Y\\ \omega\end{bmatrix}}.
\end{equation}
Here, $\bG(\zeta):= i[\bH, \zeta] + \bD(\zeta)$ is the GKSL generator of the plant-filter system, and
$
    \bD(\zeta)
    =
    -[\zeta, \bL^{\rT}]\bOmega \bL - \frac{1}{2} [\bL^{\rT}\bOmega \bL, \zeta]
$ is  the corresponding decoherence superoperator, similar to (\ref{cD}).
In (\ref{dxi}), use is also made of the partial GKSL decoherence superoperator $\Delta$  which acts on the filter variables as
$    \Delta(\xi)
    =
    -[\xi, \Phi^{\rT}]\Omega \Phi -[\xi, \Psi^{\rT}]\mho \Psi - \frac{1}{2} [\Phi^{\rT}\Omega \Phi + \Psi^{\rT}\mho\Psi, \xi]$  in view of (\ref{cX})--(\ref{bHbL}).
Now, consider a CQF problem which is formulated as the minimization of the mean square discrepancy
\begin{equation}
\label{CQF}
    \cZ
    :=
    \lim_{t\to +\infty} \bE Z(t) \longrightarrow \min,
    \qquad
    Z :=
    \frac{1}{2}
    (S X - T\bL)^{\rT}(S X - T\bL)
\end{equation}
between $r$ linear combinations of the plant variables of interest and filter output variables specified by given matrices $S\in \mR^{r\x n}$ and $T:= \begin{bmatrix}T_1 & T_2\end{bmatrix}$, with $T_1\in \mR^{r\x m}$, $T_2\in \mR^{r\x \mu}$ (note that, similarly to (\ref{dY}), $2\Im\bOmega \bL$ is the signal part of $\eta$  in view of (\ref{xiomlam}) and (\ref{cW})).
The criterion process $Z$ in (\ref{CQF}) is similar  to that in (\ref{Zex}):
\begin{equation}
\label{ZCQF}
    Z :=
    \frac{1}{2}
    {\begin{bmatrix}
        X^{\rT} & \bL^{\rT}
    \end{bmatrix}}
    \Pi
    {\begin{bmatrix}
        X\\
        \bL
    \end{bmatrix}},
    \qquad
    \Pi
    :=
    {\begin{bmatrix}
        S^{\rT} S & -S^{\rT} T \\
        -T^{\rT} S & T^{\rT}T
    \end{bmatrix}},
\end{equation}
with only the subvector $X$ of $\cX$ from (\ref{cX}) being present in (\ref{ZCQF}).  The minimization in (\ref{CQF})
is carried out over the filter Hamiltonian $\Gamma$ and the vector $\Phi$ of the filter-plant coupling operators in (\ref{PhiPsi}).
This problem extends \cite{VP_2013b} in that we do not restrict attention to linear filters even if the plant is an OQHO and $\Psi$ depends linearly on the filter variables $\xi$. If $\Gamma_0$ and $\Phi_0$ are perturbed in the directions
\begin{equation}
\label{KM0CQF}
    K_0:= \Gamma_0',
    \qquad
    M_0:= \Phi_0',
\end{equation}
consisting of self-adjoint operators representable as functions of the filter variables, the  corresponding perturbations of the plant-filter Hamiltonian and coupling operators in (\ref{bHbL}) are
\begin{equation}
\label{KMCQF}
    \bH_0' = K_0  + M_0^{\rT}JL_0,
    \qquad
    \bL_0'= {\begin{bmatrix}M_0\\ 0\end{bmatrix}}.
\end{equation}
By applying Theorem~\ref{th:Q}, and using (\ref{cW}), (\ref{KM0CQF}) and (\ref{KMCQF}),
the corresponding transverse Hamiltonian $Q$ for the plant-filter system satisfies the QSDE
\begin{align}
\nonumber
    \rd Q
    & =
        \Big(K  + M^{\rT}JL -\Im \Big(\bL^{\rT}\bOmega {\small\begin{bmatrix}M\\ 0\end{bmatrix}}\Big)\Big)\rd t
        +
        {\begin{bmatrix}M^{\rT}& 0\end{bmatrix}}\rd \cW\\
\label{dQCQF}
    & =
        \big(K   -\Im ((2L+\Phi)^{\rT}\Omega M)\big)\rd t
        +
        M^{\rT} \rd W,
\end{align}
where use is also made of $\Im(L^{\rT}\Omega M) 
= -M^{\rT}JL$ (following from (\ref{Omega}) and the commutativity $[L,M^{\rT}] = 0$). Then the Gateaux derivative $\cZ'$ of the cost functional in (\ref{CQF}) is given by (\ref{cZ'}), where in view of (\ref{ZCQF}) and (\ref{KMCQF}) (and similarly to (\ref{Z0'ex})),
$$
    \sj_t(Z_0') = \Re ((T\bL - SX)^{\rT} T_1M).
$$
In accordance with (\ref{Ephidot}), the expectation of the derivative process $\phi:= i[Q,Z]$ satisfies the IDE
$$
  (\bE \phi)^{^\centerdot} = i\bE[Q,\bG(Z)] + \bE \chi(Z),
$$
where, in view of (\ref{dQCQF}), the superoperator $\chi$ in (\ref{chi}) is given by
$$
    \chi(Z)
     =
        i[K   -\Im ((2L+\Phi)^{\rT}\Omega M),Z]
        -2\Re ([Z,(L+\Phi)^{\rT}]\Omega M).
$$
Therefore, in the case when both the plant and the unperturbed filter are OQHOs with Hurwitz dynamics matrices (while the perturbations in (\ref{KM0CQF}) are not  necessarily linear-quadratic), $\cZ'$ can be computed explicitly along the lines of Example~\ref{ex:OQHOquad} of Section~\ref{sec:func} due to $Z$ being a quadratic form of the plant-filter variables $\cX$.
The transverse Hamiltonian method, outlined above, has recently been used in \cite{V_2015b}  to show that, in the mean square optimal CQF problem for linear plants, linear filters can not be improved locally  (in the sense of first-order optimality conditions) by varying the energy operators of the filter along the Weyl operators.



\section{CONCLUSION}\label{sec:conc}

For a class of open quantum systems governed by Markovian  Hudson-Parthasarathy QSDEs, we have introduced a transverse Hamiltonian and derivative processes associated with perturbations of the energy operators. This provides a fully quantum tool for infinitesimal perturbation analysis of system operators and averaged cost functionals.  The proposed variational method can be used for the development of optimality  conditions in coherent quantum control problems. We have illustrated these ideas  for OQHOs with quadratic performance criteria and the mean square optimal CQF problem. 

\end{document}